\documentclass[11pt]{article}
\usepackage{graphicx}
\usepackage{amsfonts}
\usepackage{amsmath}
\usepackage{amsthm}
\usepackage{bm}                     
\usepackage{bbm}
\usepackage{color}
\usepackage{longtable}
\usepackage{bbm}
\usepackage[all]{xy}

\usepackage[paper=a4paper,dvips,top=3cm,left=2.4cm,right=2.4cm,
    foot=1cm,bottom=3cm]{geometry}

\begin{document}

\title{A Polynomially Irreducible Functional Basis of Hemitropic Invariants of Piezoelectric Tensors}
\author{Yannan Chen\footnote{%
    Department of Applied Mathematics, The Hong Kong Polytechnic University,
    Hung Hom, Kowloon, Hong Kong ({\tt yannan.chen@polyu.edu.hk}).
    This author was supported by the National Natural Science Foundation of China (Grant No. 11571178, 11771405).}
\and Zhenyu Ming\footnote{%
    Department of Mathematical Sciences, Tsinghua University, Beijing 100084, China. ({\tt mingzy17@mails.tsinghua.edu.cn}).
    This author was supported by the National Natural Science Foundation of China (Grant No. 11271221, 11771244).}
\and Liqun Qi\footnote{%
    Department of Applied Mathematics, The Hong Kong Polytechnic University,
    Hung Hom, Kowloon, Hong Kong ({\tt maqilq@polyu.edu.hk}).
    This author's work was partially supported by the Hong Kong Research Grant Council
    (Grant No. PolyU  15302114, 15300715, 15301716 and 15300717).}
\and Wennan Zou\footnote{%
    Institute for Advanced Study, Nanchang University, Nanchang 330031, China ({\tt zouwn@ncu.edu.cn}).
    This author was supported by the National Natural Science Foundation of China (Grant No. 11372124).}
}

\date{\today}
\maketitle

\begin{abstract}
  For piezoelectric tensors, Olive (2014) proposed a minimal integrity basis of 495 hemitropic invariants, which is also a functional basis. In this article, we construct a new functional basis of hemitropic invariants of piezoelectric tensors, using the approach of Smith and Zheng. By eliminating invariants that are polynomials in other invariants, we obtain a new functional basis with 260 polynomially irreducible hemitropic invariants. Thus, the number of hemitropic invariants in the new functional basis is substantially smaller than the number of invariants in a minimal integrity basis.

  \medskip

  \textbf{Key words.} functional basis, hemitropic invariant, piezoelectric tensor.
\end{abstract}

\newtheorem{Theorem}{Theorem}[section]
\newtheorem{Definition}[Theorem]{Definition}
\newtheorem{Lemma}[Theorem]{Lemma}
\newtheorem{Corollary}[Theorem]{Corollary}
\newtheorem{Proposition}[Theorem]{Proposition}
\newtheorem{Conjecture}[Theorem]{Conjecture}
\newtheorem{Question}[Theorem]{Question}

\newcommand{\REAL}{\mathbb{R}}
\newcommand{\COMP}{\mathbb{C}}
\newcommand{\vt}[1]{{\bf #1}}
\newcommand{\aaa}{{\vt{a}}}
\newcommand{\ddd}{{\vt{d}}}
\newcommand{\cc}{{\vt{c}}}
\newcommand{\uu}{{\vt{u}}}
\newcommand{\vv}{{\vt{v}}}
\newcommand{\ww}{{\vt{w}}}
\newcommand{\x}{{\vt{x}}}
\newcommand{\y}{{\vt{y}}}
\newcommand{\z}{{\vt{z}}}
\newcommand{\e}{{\vt{e}}}
\newcommand{\A}{{\bf A}}
\newcommand{\B}{{\bf B}}
\newcommand{\Dd}{{\bf D}}
\newcommand{\D}{{\bf D}}
\newcommand{\E}{{\bf E}}
\newcommand{\F}{{\bf F}}
\newcommand{\G}{{\bf G}}
\renewcommand{\H}{{\bf H}}
\newcommand{\K}{{\bf K}}
\newcommand{\U}{{\bf U}}
\newcommand{\V}{{\bf V}}
\newcommand{\W}{{\bf W}}
\newcommand{\Pie}{{\bf P}}
\newcommand{\SPie}{{\mathbb{P}\mathrm{iez}}}
\newcommand{\SH}[1]{{\mathbb{H}^{#1}}}
\newcommand{\OO}{{\mathrm{O}(3)}}
\newcommand{\SO}{{\mathrm{SO}(3)}}
\newcommand{\LCT}{{\bm \epsilon}}
\newcommand{\COV}{\mathbf{Cov}}
\newcommand{\INV}{\mathbf{Inv}}
\newcommand{\BS}[1]{{\mathcal{S}_{#1}}}
\newcommand{\uim}{\mathbf{\Large i}}
\renewcommand{\epsilon}{\varepsilon}
\newcommand{\tabincell}[2]{\begin{tabular}{@{}#1@{}}#2\end{tabular}}
\newcommand{\Epsilon}{{\bm\epsilon}}
\newcommand{\tr}{\mathrm{tr}}
\newcommand{\Es}{{\bf H}}
\newcommand{\HH}{{\bf H}}
\newcommand{\ee}{{\bf\Epsilon}}

\section{Introduction}

In the early 1880s, Curie brothers \cite{CC-80,CC-81} discovered the piezoelectricity in certain crystalline materials with no centrosymmetry, which describes a conversion from mechanical energy into electricity or vice-versa \cite{CJQ-17,Ha-07,ZTP-13}. The piezoelectric tensor arises form a linear electromechanical interaction and hence is a third order tensor with the last two indices symmetric in a three-dimensional physical space.
Piezoelectric tensor is one of the most important tensors which has extensive applications in physics and engineering. For instance, it has been widely used in  crystal study \cite{Ha-07,Kh-08,Lo-89,Ny-85,ZTP-13} and also been applied to production and detection of sound, generation of high voltages,
electronic frequency generation, microbalances, and ultra fine focusing of optical assemblies \cite{Kh-08}.


The theory of representations for tensor functions is also of prime importance in the rational of material behaviors \cite{Boe-87,Zh-94,TM-06}. It was introduced to describe general consistent invariant forms of the nonlinear constitutive equations and to determine the number and the type of scalar variables involved.  In the latter half of the twentieth century, representations in complete and irreducible forms of vectors, second order symmetric tensors and second order skew-symmetric tensors for both isotropic and hemitropic invariants in two- and three-dimensional physical spaces, were well established by Spencer \cite{Spe-70}, Wang \cite{W-701,W-702,W-7071}, Smith \cite{Sm-71}, Boehler \cite{Boe-77}, Pennisi and Trovato \cite{PT-87} and simplified by Zheng \cite{Zh-94}. In recent years, a series of breakthroughs for third  and fourth order tensors have been achieved in this field \cite{OA-14,Ol-17,OKA-17,CHQZ-18,CLQZZ-18,LDQZ-18}.

The piezoelectric tensor contains eighteen independent elements in a three-dimensional physical space, since the last two indices are symmetric. As a special case of piezoelectric tensors, the third order symmetric and traceless tensor has seven independent elements. Smith and Bao \cite{SB-97} gave a minimal integrity basis of 4 isotropic invariants for third order symmetric and traceless tensors. Chen, Hu, Qi and Zou \cite{CHQZ-18} proved that the Smith--Bao minimal integrity basis is also a minimal functional basis of isotropic invariants of third order symmetric and traceless tensors. The third order symmetric tensor is another spacial piezoelectric tensor and it has ten independent elements. By the recent work of Olive and Auffray \cite{OA-14}, a minimal integrity basis of third order symmetric tensors contains 13 isotropic invariants. Remarkably, Chen, Liu, Qi, Zheng and Zou \cite{CLQZZ-18} claimed that eleven out of thirteen isotropic invariants in the Olive--Auffray basis form a minimal functional basis of third order symmetric tensors. Liu, Ding, Qi and Zou \cite{LDQZ-18} proposed a set of 10 isotropic invariants which forms a minimal integrity basis as well as a minimal functional basis of third order Hall tensors. In addition, Olive, Kolev and Auffray \cite{OKA-17} presented a minimal integrity of 297 isotropic invariants for fourth order elasticity tensors.

In 2014, Olive \cite{Ol-14} presented a minimal integrity basis of hemitropic invariants of the piezoelectric tensor, which consists of $495$ hemitropic invariants.   Is it possible to find a functional basis of hemitropic invariants of the piezoelectric tensor, which consists of polynomial invariants, such that the number of hemitropic invariants in that basis is substantially smaller than $495$?   We will give a positive answer to this question in this paper.

In this article, to obtain a functional basis of piezoelectric tensors, we apply a constructive method which was developed by Smith \cite{Sm-71} and Zheng \cite{Zh-93b}. For a group of second order symmetric tensors, second order skew-symmetric tensors, and vectors, Smith \cite{Sm-71} constructed a set of invariants of the tensor group, such that all tensors in the group could be determined from these invariants under a certain orthonormal basis. Zheng \cite{Zh-93,Zh-93b} refined this method further.  The constructive method provides several intermediate tensors with order one and two. Generally speaking, for recovering a higher order piezoelectric tensor, it is better to start from intermediate tensors than to begin with only zero order tensors. This is the motivation of our paper.

By the orthogonal irreducible decomposition \cite{ZZ-01} of tensors, the piezoelectric tensor is factorized into four parts: a third order symmetric and traceless tensor, a second order symmetric and traceless tensor, and two vectors. Since functional bases of second order tensors and vectors are well-studied, the third order symmetric and traceless tensor is the only undetermined tensor. By exploring contraction of indices of different tensors, we construct nine intermediate tensors: five second order symmetric tensors and four vectors. Using the approach of Smith \cite{Sm-71} and Zheng \cite{Zh-93b}, we obtain a functional basis of 393 hemitropic invariants of these nine intermediate tensors. Next, starting from these nine intermediate tensors, we recover seven independent elements of the third order symmetric and traceless tensor under a proper orthonormal basis. This means that the functional basis of 393 hemitropic invariants of nine intermediate tensors is also a functional basis of hemitropic invariants of the piezoelectric tensor.
Finally, by eliminating hemitropic invariants that are polynomials of other invariants in the functional basis, we obtain a polynomially irreducible functional basis of piezoelectric tensors which contains 260 hemitropic invariants.


This paper is organized as follows. In Section \ref{Sec:decom}, we introduce some basic definitions of tensor spaces and the orthogonal irreducible decomposition of piezoelectric tensors. In Section \ref{Sec:Func1}, preliminary definitions of invariants and functional bases of second order tensors and vectors are presented. In Section \ref{Rec}, starting from a set of nine intermediate tensors related to a piezoelectric tensor, we prove that the piezoelectric tensor is determined by these intermediate tensors. For this reason, a functional basis of the piezoelectric tensor is equivalent to a functional basis of these intermediate tensors. In Section \ref{Sec:Func2}, we present a polynomially irreducible functional basis of 260 hemitropic invariants of piezoelectric tensors as a final result. Moreover, some special cases are considered to partially verify the correctness of our work.
Finally, some concluding remarks are addressed in Section \ref{Sec:FinRemk}.

\section{Preliminary}

In this section, we introduce some basic definitions and related results on the theory of representations for tensor functions.

\subsection{Decomposition of a piezoelectric tensor}\label{Sec:decom}

We denote $\SPie$ as the linear space of piezoelectric tensors with order three in a three-dimensional physical space.
Clearly, a piezoelectric tensor contains $18$ independent elements:
\begin{equation*}
  \begin{array}{cccccc}
    P_{111}, & P_{122}, & P_{133}, & P_{123}, & P_{113}, & P_{112}, \\
    P_{211}, & P_{222}, & P_{233}, & P_{223}, & P_{213}, & P_{212}, \\
    P_{311}, & P_{322}, & P_{333}, & P_{323}, & P_{313}, & P_{312}.
   \end{array}
\end{equation*}

Let $\mathbb{T}_{ijk}$ and $\mathbb{T}_{(ijk)}$ be the real linear space of third order tensors and the real linear space of third order symmetric tensors in a three-dimensional physical space, respectively. Here, the notation $(..)$ means invariance under all permutations of indices in parentheses. In this sense, we may denote $\SPie=\mathbb{T}_{i(jk)}$.


Given a positive oriented orthonormal basis $\{\e_1,\e_2,\e_3\}$ of the three-dimensional physical space, a tensor ${\bf T}\in\mathbb{T}_{ijk}$ could be represented as
\begin{equation*}
  {\bf T} = T_{ijk} \cdot \e_i\otimes\e_j\otimes\e_k,
\end{equation*}
where $T_{ijk}$ is a three-way array and $\otimes$ stands for the tensor product. We call $T_{ijk}$ the representing array of the tensor ${\bf T}$ and denote ${\bf T}=(T_{ijk})$ for a given coordinate system.

Let $\SO$ be the rotation group in dimension three. Under an orthonormal basis, every rotation is described by an orthogonal 3-by-3 matrix $g$ with $\det(g)=1$.
An $\SO$-action on $\mathbb{T}_{ijk}$ is denoted by $\ast$ and defined by
\begin{equation*}
  \ast : \SO\times\mathbb{T}_{ijk} \to \mathbb{T}_{ijk}; \qquad
  (g,{\bf T})\mapsto {\bf T}':=g\ast{\bf T} \text{ with }T'_{ijk}=g_{ir}g_{js}g_{kt}T_{rst}.
\end{equation*}
A subspace $\mathbb{F}\subseteq \mathbb{T}_{ijk}$ is $\SO$-stable if for all ${\bf T}\in\mathbb{F}$, it holds that
\begin{equation*}
  g\ast {\bf T}\in\mathbb{F} ~~ \forall g\in\SO.
\end{equation*}
Generally, an $\SO$-stable space may be decomposed into smaller $\SO$-stable subspaces. If a space contains no proper non-trivial $\SO$-stable subspace, we call it irreducible under the $\SO$-action.

Let $\SH{n}$ be the space of $n$th order symmetric and traceless tensors. Here, traceless means that
\begin{equation*}
  T_{iij}=T_{iji}=T_{jii}=0, ~~\forall j
\end{equation*}
provided $(T_{ijk})\in\mathbb{T}_{ijk}$. Since there is a classical isomorphism in the three-dimensional physical space between $\SH{n}$ and $n$th-degree harmonic homogeneous polynomials, a symmetric and traceless tensor is also called a harmonic tensor. All scalars and vectors are naturally harmonic. It is a classical statement that $\SH{n}$ is irreducible under $\SO$-actions \cite{OA-14}. Since $\SPie\supset\mathbb{T}_{(ijk)}\supset\SH{3}$, $\SPie$ is not irreducible.

Now, we factorize the space of piezoelectric tensors $\SPie$ into four irreducible subspaces \cite{ZZ-01}:
\begin{equation}\label{decom-piez-space}
  \SPie ~~\to~~ \SH{3}\oplus\SH{1}\oplus\SH{2}\oplus\SH{1}.
\end{equation}
Using the approach given in \cite{Spe-70}, we split a piezoelectric tensor $\Pie$ into four parts $(\A,\uu,\D,\vv)$.
The process is illustrated as follows
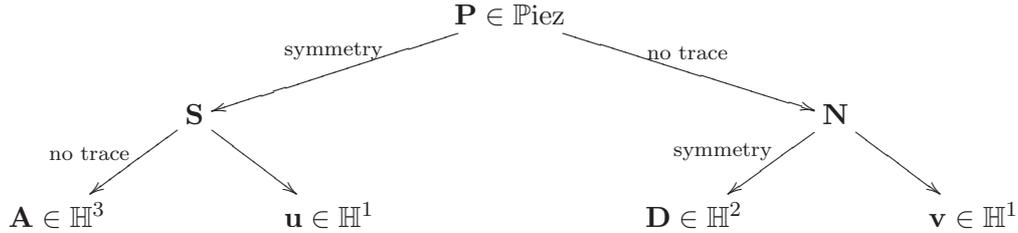
\begin{equation*}
  \xymatrix{
  &&& \Pie\in\SPie \ar[dll]_{\text{symmetry}}\ar[drr]^{\text{no trace}} &&&                          \\
  & \mathbf{S} \ar[dl]_{\text{no trace}} \ar[dr]^{} &&&& \mathbf{N} \ar[dl]_{\text{symmetry}} \ar[dr]^{} & \\
 \A\in\SH{3} && \uu\in\SH{1} && \D\in\SH{2} && \vv\in\SH{1}                  }
\end{equation*}
where $\mathbf{S}$ is a third order symmetric tensor with traces, $\mathbf{N}$ is a second order traceless tensor but it is asymmetric.

Let us see more details. We denote $\epsilon_{ijk}$ the Levi-Civita symbol and $\delta_{ij}$ the Kronecker delta:
\begin{equation*}
  \epsilon_{ijk}=\left\{\begin{array}{ll}
    1  & \text{ if }(i,j,k)\in\{(1,2,3),(2,3,1),(3,1,2)\}, \\
    -1 & \text{ if }(i,j,k)\in\{(1,3,2),(2,1,3),(3,2,1)\}, \\
    0  & \text{ otherwise, }
  \end{array}\right.
  \qquad \text{ and } \qquad
  \delta_{ij}=\left\{\begin{array}{ll}
    1  & \text{ if }i=j, \\
    0  & \text{ otherwise. }
  \end{array}\right.
\end{equation*}
For a given piezoelectric tensor $\Pie=(P_{ijk})$, we first compute a second order traceless tensor $\mathbf{N}=(N_{ij})$ and a third order symmetric tensor $\mathbf{S}=(S_{ijk})$ by
\begin{equation*}
  N_{ij}=\epsilon_{k\ell j}P_{\ell ki} \qquad\text{ and }\qquad
  S_{ijk}=P_{ijk}-\frac{1}{3}(\epsilon_{ji\ell}N_{k\ell}+\epsilon_{ki\ell}N_{j\ell}),
\end{equation*}
respectively. Second, from the second order traceless tensor $\mathbf{N}=(N_{ij})$, we calculate a vector $\vv=(v_k)$ and a second order symmetric and traceless tensor $\D=(D_{ij})$, where
\begin{equation*}
  v_k=\epsilon_{ijk}N_{ij} \qquad\text{ and }\qquad
  D_{ij}=N_{ij}-\frac{1}{2}\epsilon_{ijk}v_k.
\end{equation*}
Finally, by the harmonic decomposition of the third order symmetric tensor $\mathbf{S}=(S_{ijk})$, we obtain a vector $\uu=(u_k)$ and a third order symmetric and traceless tensor $\A=(A_{ijk})$ via
\begin{equation*}
  u_k=S_{iik} \qquad\text{ and }\qquad
  A_{ijk}=S_{ijk}-\frac{1}{5}(u_i\delta_{jk}+u_j\delta_{ik}+u_k\delta_{ij}).
\end{equation*}

Next, we address existing results on hemitropic invariants of lower order tensors.

\subsection{A functional basis of hemitropic invariants of second-order symmetric tensors and vectors}\label{Sec:Func1}

Before we start, we give some preliminary definitions.
If for all ${\bf T}\in\mathbb{T}_{ijk}$, a scalar-valued function $I(\cdot)$ satisfies
\begin{equation*}
  I({\bf T}) = I(g\ast{\bf T}) ~~\forall g\in\SO,
\end{equation*}
we call $I$ a hemitropic invariant of $\mathbb{T}_{ijk}$.
When we restrict scalar-valued functions in homogeneous polynomials, the algebra of invariant polynomials on $\mathbb{T}_{ijk}$ is finitely generated, according to invariant theory \cite{Hi-93}.

\begin{Definition}
  Let $\{I_1,I_2,\dots,I_r\}$ be a finite set of hemitropic invariants of $\mathbb{T}_{ijk}$ that are all homogeneous polynomials. If all hemitropic invariant polynomials of $\mathbb{T}_{ijk}$ are polynomials in $I_1,I_2,\dots,I_r$, we call the set $\{I_1,I_2,\dots,I_r\}$ an integrity basis of $\mathbb{T}_{ijk}$. An integrity basis is minimal if none of its proper subset is an integrity basis.
\end{Definition}

If we relax invariants from polynomials to scalar-valued functions, we get the functional basis.

\begin{Definition}
  A finite set of hemitropic invariants $\{I_1,I_2,\dots,I_r\}$ of $\mathbb{T}_{ijk}$ is called a functional basis of $\mathbb{T}_{ijk}$ if
  \begin{equation*}
    I_i({\bf T}_1) = I_i({\bf T}_2) ~~\forall i=1,\dots,r
  \end{equation*}
  imply ${\bf T}_1 = g\ast {\bf T}_2$ for some $g\in\SO$. A functional basis is minimal if none of its proper subset is a functional basis.
\end{Definition}

For a given tensor ${\bf T}$, a set of tensors
\begin{equation*}
  \SO\ast{\bf T} = \{g\ast{\bf T}:g\in\SO\}
\end{equation*}
is called the $\SO$-orbit of ${\bf T}$. Clearly, the functional basis has the property of separating $\SO$-orbits \cite{OKA-17}.
Since integrity bases are also functional bases \cite{BKO-94}, both integrity bases and functional bases could separate orbits.
In a geometric viewpoint, a piezoelectric material is a point in the orbit space $\SPie/\SO$.

Once a typical tensor ${\bf T}_1$ in the $\SO$-orbit $(\SO\ast{\bf T}_1)$ is determined by the set of hemitropic invariants $\{I_1,I_2,\dots,I_r\}$, we compute any hemitropic invariant of tensor in $(\SO\ast{\bf T}_1)$ from ${\bf T}_1$ directly, which is clearly a scalar-valued function in $I_1,I_2,\dots,I_r$. Thus, the set $\{I_1,I_2,\dots,I_r\}$ is named a functional basis of a tensor space.

Smith \cite{Sm-71} proposed a constructive approach for determining a functional basis of second-order symmetric tensors $\A_1,\dots,\A_M$, second order skew-symmetric tensors $\W_1,\dots,\W_N$, and first order vectors $\vv_1,\dots,\vv_P$ in a common coordinate system. Whereafter, Boehler \cite{Boe-77} refined Smith's functional basis. Pennisi and Trovato \cite{PT-87} proved that the refined functional basis is minimal. A summarize work was due to Zheng \cite{Zh-94}.

Let $\{\e_1,\e_2,\e_3\}$ be a positive oriented orthonormal basis. The Levi-Civita tensor $\Epsilon=\epsilon_{ijk}\e_i\otimes\e_j\otimes\e_k$ is a constant tensor under the $\SO$-action. Hence, there is a one-to-one correspondence between skew-symmetric tensors $\W=(W_{ij})$ and its axial vectors $\vv=(v_i)$ \cite{Zh-93b}:
\begin{equation*}
  \W = - \Epsilon\vv \qquad\text{ and }\qquad \vv=-\frac{1}{2}\Epsilon[\W],
\end{equation*}
where $\Epsilon\vv:=\epsilon_{ijk}v_k \e_i\otimes\e_j$ and $\Epsilon[\W]:=\epsilon_{ijk}W_{jk}\e_i$. Here, we only consider the functional basis of second-order symmetric tensors $\A_1,\dots,\A_M$ and vectors $\vv_1,\dots,\vv_P$, which contains the following hemitropic invariants \cite{Zh-93b,Zh-94}:
\begin{equation}\label{SmithBasis}
\left\{\begin{aligned}
  & \vv_{\alpha}\cdot\vv_{\alpha},~~ \vv_{\alpha}\cdot\vv_{\beta},~~ [\vv_{\alpha},\vv_{\beta},\vv_{\gamma}], \\
  & \mathrm{tr} \A_{\mu},~~ \mathrm{tr} \A_{\mu}^2,~~ \mathrm{tr} \A_{\mu}^3,~~
    \mathrm{tr} \A_{\mu}\A_{\nu},~~ \mathrm{tr} \A_{\mu}^2\A_{\nu},~~ \mathrm{tr} \A_{\mu}\A_{\nu}^2,~~ \mathrm{tr} \A_{\mu}^2\A_{\nu}^2,~~ \mathrm{tr} \A_{\mu}\A_{\nu}\A_{\sigma}, \\
  & \vv_{\alpha}\cdot \A_{\mu}\vv_{\alpha},~~ \vv_{\alpha}\cdot \A_{\mu}^2\vv_{\alpha},~~ [\vv_{\alpha},\A_{\mu}\vv_{\alpha},\A_{\mu}^2\vv_{\alpha}], \\
  & \vv_{\alpha}\cdot \Epsilon[\A_{\mu}\A_{\nu}],~~ \vv_{\alpha}\cdot \Epsilon[\A_{\mu}^2\A_{\nu}],~~ \vv_{\alpha}\cdot \Epsilon[\A_{\mu}\A_{\nu}^2],~~
    [\vv_{\alpha},\A_{\mu}\vv_{\alpha},\A_{\nu}\vv_{\alpha}], \\
  & \vv_{\alpha}\cdot \A_{\mu}\vv_{\beta},~~ [\vv_{\alpha},\vv_{\beta},\A_{\mu}\vv_{\alpha}],~~ [\vv_{\alpha},\vv_{\beta},\A_{\mu}\vv_{\beta}],
\end{aligned}\right.
\end{equation}
where $\alpha,\beta,\gamma\in\{1,2,\dots,P\}$ with $\alpha<\beta<\gamma$, $\mu,\nu,\sigma\in\{1,2,\dots,M\}$ with $\mu<\nu<\sigma$, and $[\uu,\vv,{\bf w}]=\vv\cdot(\Epsilon\uu){\bf w}$ is the scalar triple product.

\section{Recovery of a piezoelectric tensor}\label{Rec}


According to the decomposition of piezoelectric tensors \eqref{decom-piez-space}, we know
\begin{equation}\label{piezo-decom}
  P_{ijk} = A_{ijk}
              + \frac{1}{3}(\epsilon_{i\ell k}D_{\ell j}+\epsilon_{i\ell j}D_{\ell k})
              + \frac{1}{5}(\delta_{ij}u_k+\delta_{ik}u_j+\delta_{jk}u_i)
              + \frac{1}{6}(\epsilon_{ij\ell}\epsilon_{\ell km}v_m+\epsilon_{i\ell k}\epsilon_{\ell mj}v_m),
\end{equation}
where $\A=(A_{ijk})\in\SH{3}, \D=(D_{ij})\in\SH{2},$ and $\uu=(u_i),\vv=(v_i)\in\SH{1}$.
For convenience, we define some tensors:
\begin{eqnarray*}
  & \B:=A_{ik\ell}A_{jk\ell}\e_i\otimes\e_j, \quad \cc:=A_{ijk}B_{jk}\e_i, \quad
    \F:=A_{ijk}u_k\e_i\otimes\e_j, \quad \G:=A_{ijk}v_k\e_i\otimes\e_j, & \\
  & \E:=A_{ik\ell}\epsilon_{jm\ell}D_{km}\e_i\otimes\e_j, \quad
    \ww:=-\tfrac{1}{2}\Epsilon[\E], \qquad \H:=\E+\Epsilon\ww,
    &
\end{eqnarray*}
where $\B=(B_{ij})$ is a second order symmetric tensor, $\F=(F_{ij})$ and $\G=(G_{ij})$ are second order symmetric and traceless tensors, $\E=(E_{ij})$ is a second order traceless and asymmetric tensor, which is a sum of a symmetric and traceless tensor $\H$ and a skew-symmetric tensor $(-\Epsilon\ww)$, and $\cc=(c_i)$ and $\ww$ are vectors. Clearly, $\E$ is equivalent to $\H$ and $\ww$.

The outline of the process for constructing a functional basis of piezoelectric tensors is as follows.
Using the approach of Smith \cite{Sm-71} and Zheng \cite{Zh-93b}, we estimate a group of nine intermediate tensors
\begin{equation}\label{med-tens}
  \cc, \quad \uu, \quad \vv, \quad \ww, \quad
  \B, \quad \D, \quad \F, \quad \G,~~~ \text{and } \H
%
\end{equation}
from a set of hemitropic invariants (which will be addressed in Section \ref{Sec:Func2}) under a certain positive oriented orthonormal basis.
In a certain $\SO$-orbit, there are infinitely many group of tensors which are equivalent, so we only need to determine any one of them, i.e., a typical group of tensors is servable for identifying the $\SO$-orbit.
The key point is to deal with a group of tensors under a common positive oriented orthonormal basis.
We may further rotate tensors in the group \eqref{med-tens} simultaneously to recover the third order symmetric and traceless tensor $\A$.
Once $\D,\uu,\vv$, and $\A$ are determined in a common positive oriented orthonormal basis, the piezoelectric tensor $\Pie$ is computed from \eqref{piezo-decom} straightforwardly and hence the set of hemitropic invariants is a functional basis of piezoelectric tensors.

In the remainder of this section, we focus on the recovery of the only undetermined tensor $\A=(A_{ijk})$ provided that tensors in the group \eqref{med-tens} are known under a proper positive oriented orthonormal basis.



Smith's method is our fundamental tool for recovering a piezoelectric tensor from a set of hemitropic invariants, i.e., its functional basis.
A valuable tool in Smith's method is the composition of rotations such that multiple tensors have better structure. For example, we consider two nonzero vectors $\uu$ and $\vv$ that are not collinear. At the first step, we may rotate the coordinate system such that the direction of 1-axis is along with the vector $\uu$. Hence, we have $\uu=u_1\e_1$ where $u_1=\sqrt{\uu\cdot\uu}$. Second, since 2- and 3- components of $\uu$ are all zeros, we could fix 1-axis and further rotate 2- and 3-axes of the coordinate system such that $\vv=v_1\e_1+v_2\e_2$ with $v_2>0$, while keeping $\uu=u_1\e_1$. It is well-known that the composition of two rotations is still a rotation. In a word, we say that, under a proper positive oriented orthonormal basis $\{\e_1,\e_2,\e_3\}$, two vectors $\uu$ and $\vv$ could be represented as
\begin{equation*}
  \uu=u_1\e_1 \qquad\text{ and }\qquad \vv=v_1\e_1+v_2\e_2,
\end{equation*}
respectively. If $\uu$ and $\vv$ are not collinear, we can further obtain $u_1>0$ and $v_2>0$.

Furthermore, we define $g(\theta)$ as a rotation in the 2-3 plane with a representing array
\begin{equation*}
  (g(\theta))_{ij} = \left(\begin{array}{ccc}
    1 & 0 & 0 \\
    0 & \cos\theta & -\sin\theta \\
    0 & \sin\theta & \cos\theta \\
  \end{array}\right).
\end{equation*}
As mentioned earlier, we have $g(\theta)\ast \e_1=\e_1$ for all $\theta$ and hence $g(\theta)\ast (\e_1\otimes\e_1)=\e_1\otimes\e_1$. By linear algebra, the second order symmetric and traceless tensor
\begin{equation}\label{d_0}
  {\bf d}_0=-2\e_1\otimes\e_1+\e_2\otimes\e_2+\e_3\otimes\e_3
\end{equation}
satisfies $g(\theta)\ast {\bf d}_0= {\bf d}_0$ for all $\theta$. The skew-symmetric tensor $(\e_2\otimes\e_3-\e_3\otimes\e_2)$ also satisfies $g(\theta)\ast (\e_2\otimes\e_3-\e_3\otimes\e_2)= \e_2\otimes\e_3-\e_3\otimes\e_2$ for all $\theta$.

Now, we consider third order symmetric and traceless tensors
\begin{equation}\label{t_0}
\begin{aligned}
  {\bf d}_1(\alpha,\beta,\gamma) =&~ \gamma(-2\e_1\otimes\e_1\otimes\e_1
        +\e_1\otimes\e_2\otimes\e_2+\e_2\otimes\e_1\otimes\e_2+\e_2\otimes\e_2\otimes\e_1 \\
    &{}~~~~ +\e_1\otimes\e_3\otimes\e_3+\e_3\otimes\e_1\otimes\e_3+\e_3\otimes\e_3\otimes\e_1) \\
    &{} +\alpha(\e_2\otimes\e_2\otimes\e_2-\e_2\otimes\e_3\otimes\e_3-\e_3\otimes\e_2\otimes\e_3-\e_3\otimes\e_3\otimes\e_2) \\
    &{} +\beta(\e_2\otimes\e_2\otimes\e_3+\e_2\otimes\e_3\otimes\e_2+\e_3\otimes\e_2\otimes\e_2-\e_3\otimes\e_3\otimes\e_3).
\end{aligned}
\end{equation}
By calculations, it yields that
\begin{equation}\label{aoei}
  g(\theta)\ast {\bf d}_1(\alpha,\beta,\gamma) = {\bf d}_1(\widetilde{\alpha},\widetilde{\beta},\gamma) \qquad\text{ and }\qquad
  \alpha^2+\beta^2 = \widetilde{\alpha}^2+\widetilde{\beta}^2,
\end{equation}
where $\widetilde{\alpha}=\alpha\cos3\theta-\beta\sin3\theta$ and $\widetilde{\beta}=\alpha\sin3\theta+\beta\cos3\theta$.
Clearly, we have $g(\theta)\ast {\bf d}_1(0,0,\gamma)= {\bf d}_1(0,0,\gamma)$ for all $\theta$ and $\gamma$.
We note that patterns $\e_1$, $\e_1\otimes\e_1$, $\e_2\otimes\e_3-\e_3\otimes\e_2$, ${\bf d}_0$, and ${\bf d}_1$ are useful for the following analysis on recovering the tensor $\A$.


The third order symmetric and traceless tensor $\A$ has a representing array $A_{ijk}:${\footnotesize%
\begin{equation*}
  \left(\begin{array}{ccc|ccc|ccc}
    A_{111} & A_{112} & A_{113}          & A_{112} & A_{122} & A_{123}          & A_{113} & A_{123} & -A_{111}-A_{122} \\
    A_{112} & A_{122} & A_{123}          & A_{122} & A_{222} & A_{223}          & A_{123} & A_{223} & -A_{112}-A_{222} \\
    A_{113} & A_{123} & -A_{111}-A_{122} & A_{123} & A_{223} & -A_{112}-A_{222} & -A_{111}-A_{122} & -A_{112}-A_{222} & -A_{113}-A_{233}
  \end{array}\right),
\end{equation*}
}which has seven independent elements $A_{111},A_{122},A_{112},A_{222},A_{113},A_{223},$ and $A_{123}$. To determine these elements, we consider the following cases. Before we start, we give two propositions.

\begin{Proposition}\label{Prop-1}
  If $I_2:=\mathrm{tr}\B=0$ or $I_4:=\mathrm{tr}\B^2=0$, $A_{ijk}$ is a zero tensor.
\end{Proposition}
\begin{proof}
  It is straightforward to see that $A_{ijk}=0$ for all $i,j,$ and $k$ if $I_2=A_{ijk}A_{ijk}=0$.
  On the other hand, if $I_4=B_{ij}B_{ij}=0$, we have $B_{ij}=0$ for all $i$ and $j$. Thus, $B_{ii}=I_2=0$ and hence $A_{ijk}=0$ for all $i,j,$ and $k$.
\end{proof}

\begin{Proposition}\label{Prop-2}
  Let $\gamma,\zeta$, and $\Delta\ge0$ be constants and let $\alpha$ and $\beta$ be two undetermined parameters.
  Suppose vectors $\uu$, and $\vv$ are parallel to $\e_1$, $\D=\zeta{\bf d}_0$ and $\A={\bf d}_1(\alpha,\beta,\gamma)$ with $\alpha^2+\beta^2=\Delta$,  under a positive oriented orthonormal basis $\{\e_1,\e_2,\e_3\}$.
  Then, for all $\alpha$ and $\beta$ satisfying $\alpha^2+\beta^2=\Delta$, tensor groups $(\A={\bf d}_1(\alpha,\beta,\gamma),\D,\uu,\vv)$ are in the same $\SO$-orbit with a typical group $({\bf d}_1(\sqrt{\Delta},0,\gamma),\D,\uu,\vv)$.
\end{Proposition}
\begin{proof}
  Denote $\uu=u_1\e_1$ and $\vv=v_1\e_1$. By direct computations, we have
  \begin{equation*}
  \left\{\begin{aligned}
  & \B = 2(\gamma^2+\Delta){\bf d}_0 + (10\gamma^2+4\Delta)\e_1\otimes\e_1, \\
  & \E = -3\gamma\zeta(\e_2\otimes\e_3-\e_3\otimes\e_2), \\
  & \F = u_1\gamma{\bf d}_0, \\
  & \G = v_1\gamma{\bf d}_0, \\
  & \cc = 4\gamma(\Delta-2\gamma^2)\e_1.
  \end{aligned}\right.
  \end{equation*}
  Clearly, when we fix 1-axis and rotate 2- and 3-axes of the positive oriented orthonormal basis $\{\e_1,\e_2,\e_3\}$, tensors $\B,\D,\E,\F,\G,\cc,\uu,\vv$ are invariant. For the tensor $\A$, on one hand, under these rotation in the 2-3 plane, the rotated tensor $\widetilde{\A}$ could also be represented by the pattern
  \begin{equation*}
    \widetilde{\A} = {\bf d}_1(\widetilde{\alpha},\widetilde{\beta},\gamma),
  \end{equation*}
  where
  \begin{equation}\label{AAA-9}
    \widetilde{\alpha}^2+\widetilde{\beta}^2 = \Delta.
  \end{equation}
  On the other hand, all possible tensors ${\bf d}_1(\widetilde{\alpha},\widetilde{\beta},\gamma)$ satisfying \eqref{AAA-9} are located in the same $\SO$-orbit and hence are equivalence. In this sense, the $\SO$-orbit of tensor $\A$ is determined. For simplicity, we set $\A={\bf d}_1(\sqrt{\Delta},0,\gamma)$ in a typical tensor group.
\end{proof}

If $\mathrm{tr}\B=0$, we have $\A=0$ by Proposition \ref{Prop-1}. It is straightforward to construct the piezoelectric tensor from \eqref{piezo-decom}. In the following analysis, we suppose $\A\ne0$.

\medskip

\textbf{Case I: Vectors $\cc,\uu,$ and $\vv$ are not collinear.} At the beginning, we introduce a tensor
\begin{equation}\label{OP1a}
  \K := A_{ijk}c_k\e_i\otimes\e_j = \left(2B_{i\ell}B_{\ell j}-I_2B_{ij}-\frac{2I_4-I_2^2}{3}\delta_{ij}\right)\e_i\otimes\e_j,
\end{equation}
where $I_2=\tr \B$ and $I_4=\tr \B^2$ are defined in Proposition \ref{Prop-1} and the last equality is verified directly by computation. Clearly, $\K=(K_{ij})$ is completely determined by $\B$. Next, we use equations $A_{ijk}c_k=K_{ij}$, $A_{ijk}u_k=F_{ij}$, and $A_{ijk}v_k=G_{ij}$ for recovering elements of $A_{ijk}$.

Without loss of generality, we assume that two nonzero vectors $\cc$ and $\uu$ are not collinear. By choosing a positive oriented orthonormal basis $\{\e_1,\e_2,\e_3\}$, we have
\begin{equation*}
  \cc = c_1\e_1 \qquad\text{ and }\qquad \uu = u_1\e_1+u_1\e_2,
\end{equation*}
where $c_1=\sqrt{\cc\cdot\cc}>0$ and $u_2>0$. Recalling equations $A_{ijk}c_k=K_{ij}$ and $A_{ijk}u_k=F_{ij}$, we have
\begin{equation*}
\left\{\begin{aligned}
  & A_{111} = \frac{1}{c_1}K_{11}, \\
  & A_{112} = \frac{1}{c_1}K_{12}, \\
  & A_{113} = \frac{1}{c_1}K_{13}, \\
  & A_{122} = \frac{1}{c_1}K_{22}, \\
  & A_{123} = \frac{1}{c_1}K_{23}, \\
  & A_{222} = \frac{1}{u_2}(F_{22}-A_{122}u_1), \\
  & A_{223} = \frac{1}{u_2}(F_{23}-A_{123}u_1).
\end{aligned}\right.
\end{equation*}
All elements of $A_{ijk}$ are known and hence the third order symmetric and traceless tensor $\A$ is determined under the basis $\{\e_1,\e_2,\e_3\}$.

By a similar discussion, when $\cc$ and $\vv$ (resp. $\uu$ and $\vv$) are not collinear, we use $A_{ijk}c_k=K_{ij}$ and $A_{ijk}v_k=G_{ij}$ (resp. $A_{ijk}u_k=F_{ij}$ and $A_{ijk}v_k=G_{ij}$) to determine $\A$.

In the remainder two cases II and III, we suppose that vectors $\uu,\vv,$ and $\vt{c}$ are collinear.

\medskip

\textbf{Case II: $\D=0$.} 

Case II.1: Vectors $\uu,\vv,$ and $\vt{c}$ are not all zero vectors. Since $\uu,\vv,$ and $\vt{c}$ are collinear, we choose a proper positive oriented orthonormal basis $\{\e_1,\e_2,\e_3\}$ such that $\uu=u_1\e_1, \vv=v_1\e_1, \cc=c_1\e_1$ and the representing array of $\B$ has the form
\begin{equation*}
  B_{ij}=\left(\begin{array}{ccc}
    B_{11} & B_{12} & B_{13} \\ B_{12} & B_{22} & 0 \\ B_{13} & 0 & B_{33} \\
  \end{array}\right).
\end{equation*}.

If $\uu\ne0$, we solve $A_{ijk}u_k=F_{ij}$ and obtain five elements of $A_{ijk}$:
\begin{equation*}
\left\{\begin{aligned}
  & A_{111} = \frac{1}{u_1}F_{11}, \\
  & A_{112} = \frac{1}{u_1}F_{12}, \\
  & A_{113} = \frac{1}{u_1}F_{13}, \\
  & A_{122} = \frac{1}{u_1}F_{22}, \\
  & A_{123} = \frac{1}{u_1}F_{23}.
\end{aligned}\right.
\end{equation*}
In a similar way, we compute $A_{111},A_{112},A_{113},A_{122},$ and $A_{123}$ from $A_{ijk}v_k=G_{ij}$ and $A_{ijk}c_k=K_{ij}$ if $\vv\ne0$ and $\cc\ne0$, respectively.

Case II.1.1: $B_{22}\ne B_{33}$. From $A_{ijk}B_{jk}=c_i=0$ for $i\in\{2,3\}$, we have
\begin{equation*}
\left\{\begin{aligned}
  & (B_{22}-B_{33})A_{222} = -A_{112}(B_{11}-B_{33})-2A_{122}B_{12}-2A_{123}B_{13}, \\
  & (B_{22}-B_{33})A_{223} = -A_{113}(B_{11}-B_{33})-2A_{123}B_{12}+2(A_{111}+A_{122})B_{13}.
\end{aligned}\right.
\end{equation*}
Hence, we get $A_{222}$ and $A_{223}$ immediately.

Case II.1.2: $B_{22}=B_{33}$. Since $A_{ik\ell}A_{jk\ell}=B_{ij}$, combining $B_{23}=0$ and $B_{22}-B_{33}=0$, we establish a linear system
\begin{equation}\label{AAA-4}
\left\{\begin{aligned}
  & A_{113}A_{222}-A_{112}A_{223} = 2A_{111}A_{123}-2A_{112}A_{113},  \\
  & A_{112}A_{222}+A_{113}A_{223} = -A_{111}^2-2A_{111}A_{122}-A_{113}^2.
\end{aligned}\right.
\end{equation}
The determinant of this linear system is obviously $A_{112}^2+A_{113}^2\ge0$.

Case II.1.2.1: $A_{112}\ne0$ or $A_{113}\ne0$. Clearly, we solve $A_{222}$ and $A_{223}$ from the system \eqref{AAA-4} straightforwardly.

Case II.1.2.2: $A_{112}=A_{113}=0$. The linear system \eqref{AAA-4} reduces to
\begin{equation}\label{AAA-5}
\left\{\begin{aligned}
  & 2A_{111}A_{123}=0,  \\
  & A_{111}(A_{111}+2A_{122})=0.
\end{aligned}\right.
\end{equation}

Case II.1.2.2.1: If $A_{111}\ne0$. From the system \eqref{AAA-5}, we have
\begin{equation*}
  A_{123}=0 \qquad\text{ and }\qquad A_{111}=-2A_{122}.
\end{equation*}


By examining the equation $A_{2jk}A_{2jk}=B_{22}$, we find
\begin{equation}\label{AAA-8}
  A_{222}^2+A_{223}^2 = \frac{1}{2}B_{22}-A_{122}^2.
\end{equation}
Owing to $A_{112}=A_{113}=A_{123}=0$ and $A_{111}=-2A_{122}$, the tensor $\A$ satisfies the pattern $\A={\bf d}_1(A_{222},A_{223},A_{122})$.
Furthermore, $\A={\bf d}_1(A_{222},A_{223},A_{122})$ with \eqref{AAA-8}, $\D=0$, $\uu=u_1\e_1$, and $\vv=v_1\e_1$ satisfy assumptions of Proposition \ref{Prop-2}. Hence, for simplicity, we set
\begin{equation*}
  A_{222} = \sqrt{\frac{1}{2}B_{22}-A_{122}^2} \qquad\text{ and }\qquad A_{223}=0.
\end{equation*}

Case II.1.2.2.2: $A_{111}=0$. Equations $A_{ik\ell}A_{jk\ell}=B_{ij}$ reduce to
\begin{equation}\label{abcdef}
\left\{\begin{aligned}
  & A_{122}^2+A_{123}^2 = \frac{1}{2}B_{11}, \\
  & A_{222}^2+A_{223}^2 = \frac{1}{2}(B_{22}-B_{11}), \\
  & A_{223}A_{122}-A_{222}A_{123} = \frac{1}{2}B_{13}, \\
  & A_{222}A_{122}+A_{223}A_{123} = \frac{1}{2} B_{12}.
\end{aligned}\right.
\end{equation}

We claim that $B_{12}=B_{13}=0$. Otherwise, we assume $B_{12}^2+B_{13}^2>0$ for contradiction. Since $A_{111}=A_{112}=A_{113}=0$ and $B_{22}=B_{33}$, equations $A_{ijk}B_{jk}=c_i$ for $i\in\{2,3\}$ are indeed
\begin{equation*}
\left\{\begin{aligned}
  & 2B_{12}A_{122}+2B_{13}A_{123} =0, \\
  & -2B_{13}A_{122}+2B_{12}A_{123} =0.
\end{aligned}\right.
\end{equation*}
Clearly, the determinant of this linear system is positive. Hence $A_{122}=A_{123}=0$, which contradicts the last two equations in \eqref{abcdef}.

Clearly, when $B_{11}=B_{22}$, we get $A_{222}=A_{223}=0$ by solving the second equation in \eqref{abcdef}.

Then, we consider the case $B_{11}\ne B_{22}$. By the second equation in \eqref{abcdef}, we have $A_{222}^2+A_{223}^2>0$, which is the determinant of the last two linear equations in \eqref{abcdef}. Since $B_{13}=B_{12}=0$, we get $A_{122}=A_{123}=0$. By now, we know $A_{111}=A_{112}=A_{113}=A_{122}=A_{123}=0$ and hence $\A={\bf d}_1(A_{222},A_{223},0)$. Furthermore, $B_{11}=0$, $B_{22}>0$, and
\begin{equation}\label{AAA-88}
  A_{222}^2+A_{223}^2 = \frac{1}{2}B_{22}.
\end{equation}
Clearly, $\A={\bf d}_1(A_{222},A_{223},0)$ with \eqref{AAA-88}, $\D=0$, $\uu=u_1\e_1$, and $\vv=v_1\e_1$ satisfy assumptions of Proposition \ref{Prop-2}. Hence, for convenience, we set
\begin{equation*}
  A_{222} = \sqrt{\frac{1}{2}B_{22}} \qquad\text{ and }\qquad A_{223}=0.
\end{equation*}

Case II.2: $\uu=\vv=\cc=0$. Hence, tensors $\D,\E,\F,\G,$ and $\K$ are zeros. From \eqref{OP1a}, we get
\begin{equation*}
  K_{ij}=2B_{i\ell}B_{\ell j}-I_2B_{ij}-\frac{2I_4-I_2^2}{3}\delta_{ij}=0   \qquad \forall i,j.
\end{equation*}
We now choose a proper positive oriented orthonormal basis $\{\e_1,\e_2,\e_3\}$ such that the representing matrix of $\B$ is diagonal.
Clearly, the above equations with $i\ne j$ are trivial. Then, we consider the above equation with $i=j$ and obtain
\begin{equation*}
  2B_{\underline{ii}}^2-I_2B_{\underline{ii}}-\frac{2I_4-I_2^2}{3}=0  \qquad \forall i=1,2,3,
\end{equation*}
where the repeated subscript $i$ with underline is not summarized. Hence, we claim that three diagonal elements $B_{11}$, $B_{22}$, and $B_{33}$ are all roots of a quadratic equation
\begin{equation}\label{z1}
  2x^2-I_2x-\frac{2I_4-I_2^2}{3} = 0.
\end{equation}
Hence, at least two diagonal elements of $B_{ij}$ are equivalent. Thus, we assume
\begin{equation*}
  \B= B_{11}\e_1\otimes\e_1 + B_{22}(\e_2\otimes\e_2+\e_3\otimes\e_3).
\end{equation*}

Case II.2.1: $B_{11}\ne B_{22}$. By $A_{ijk}B_{jk}=c_i=0$ that are
\begin{equation*}
\left\{\begin{aligned}
  & (B_{11}-B_{22})A_{111}=0, \\
  & (B_{11}-B_{22})A_{112}=0, \\
  & (B_{11}-B_{22})A_{113}=0, \\
\end{aligned}\right.
\end{equation*}
we immediately have $A_{111}=A_{112}=A_{113}=0$. Equations $A_{ik\ell}A_{jk\ell}=B_{ij}$ reduce to
\begin{equation*}
\left\{\begin{aligned}
  & A_{122}^2+A_{123}^2=\frac{1}{2}B_{11}, \\
  & A_{222}^2+A_{223}^2=\frac{1}{2}(B_{22}-B_{11}), \\
  & A_{223}A_{122}-A_{222}A_{123} = 0, \\
  & A_{222}A_{122}+A_{223}A_{123} = 0.
\end{aligned}\right.
\end{equation*}
Since $B_{11}\ne B_{22}$, the determinant $A_{222}^2+A_{223}^2$ of the last two linear equations are nonzero. Hence, we have
\begin{equation*}
  A_{122}=A_{123}=0,
\end{equation*}
$B_{11}=0$ and $B_{22}>0$.
Hence, $\A={\bf d}_1(A_{222},A_{223},0)$ with $A_{222}^2+A_{223}^2=\frac{1}{2}B_{22}$, $\D=0$, $\uu=0$, and $\vv=0$ satisfy assumptions of Proposition \ref{Prop-2}. Hence, for convenience, we set
\begin{equation*}
  A_{222} = \sqrt{\frac{1}{2}B_{22}} \qquad\text{ and }\qquad A_{223}=0.
\end{equation*}

Case II.2.2: $B_{11}=B_{22}$ and hence $B_{ij}=B_{11}\delta_{ij}$. We could choose the positive oriented orthonormal basis $\{\e_1,\e_2,\e_3\}$ freely. Let $\e_1$ be the maximizer of $f(x)=A_{ijk}x_ix_jx_k$ on the unit sphere $\{(x_1,x_2,x_3):x_1^2+x_2^2+x_3^2=1\}$ and let $\e_2$ be the maximizer of $f(x)=A_{ijk}x_ix_jx_k$ in the circle $\{(0,x_2,x_3):x_2^2+x_3^2=1\}$. By KKT condition and some calculations, we have
\begin{equation*}
  A_{112}=A_{113}=A_{223}=0, \qquad \text{ and } \qquad A_{111}\ge A_{222}\ge 0.
\end{equation*}

We suppose $A_{111}>0$; Otherwise $A_{ijk}$ is a zero array.
By the equation $A_{2jk}A_{3jk}=B_{23}$, we immediately have $$ A_{123}=0. $$
Equation $A_{1jk}A_{2jk}=B_{12}$ reduces to $A_{222}(A_{111}+2A_{122})=0$.

(a) Assume $A_{222}=0$. Equations $A_{ik\ell}A_{jk\ell}=B_{ij}$ with $i=j$ reduce to
$$ A_{111}^2+A_{111}A_{122}+A_{122}^2 = A_{122}^2 = A_{111}^2+2A_{111}A_{122}+A_{122}^2 = \frac{1}{2}B_{11}. $$
By subtraction, we have $A_{111}A_{122}=0$ and hence $A_{111}^2=0$, which contradicts the assumption $A_{111}>0$.

(b) Hence, $A_{122}=-\frac{1}{2}A_{111}$. Equations $A_{ik\ell}A_{jk\ell}=B_{ij}$ with $i=j=1,2$ reduce to
\begin{equation*}
  \frac{3}{4}A_{111}^2 = \frac{1}{4}A_{111}^2+A_{222}^2 = \frac{1}{2}B_{11}.
\end{equation*}
Hence, by $A_{111}\ge A_{222}\ge 0$, we get
\begin{equation*}
  A_{111}=\sqrt{\frac{2B_{11}}{3}}, \quad A_{222}=\sqrt{\frac{B_{11}}{3}}, \quad \text{ and } \quad A_{122}=-\sqrt{\frac{B_{11}}{6}}.
\end{equation*}

{\sl Remark.} We can process Case II.2 by introducing the characteristic polynomial of $\B$. Because $B_{ij}$ is a diagonal matrix, its diagonal elements are all eigenvalues of $\B$. Thus, we consider the characteristic polynomial of $\B$ which is a cubic function. Applying the Cayley-Hamiltom theorem for a 3-by-3 tensor $\B$, we get
\begin{equation*}
  \B^3 -(\tr\B)\B^2+\frac{1}{2}\left((\tr\B)^2-\tr\B^2\right)\B-\det(\B){\bf I}=0,
\end{equation*}
where ${\bf I}=\delta_{ij}\e_i\otimes\e_j$ is an identity tensor. By taking the trace operation, it yields that
\begin{equation*}
  \det(\B)=\frac{1}{6}\left( (\tr\B)^3-3\tr\B^2\tr\B+2\tr\B^3 \right).
\end{equation*}
Recalling $B_{ij}=A_{ik\ell}A_{jk\ell}$ and $c_i=A_{ijk}B_{jk}=0$, we have
\begin{equation*}
  \tr\B^3=-\frac{1}{6}\left( (\tr\B)^3-5\tr\B^2\tr\B-3\cc\cdot\cc \right)
         =-\frac{1}{6}\left( (\tr\B)^3-5\tr\B^2\tr\B \right).
\end{equation*}
Hence, combining the above three equations, we obtain
\begin{equation*}
  \B^3 -(\tr\B)\B^2+\frac{1}{2}\left((\tr\B)^2-\tr\B^2\right)\B-\frac{1}{9}\left( (\tr\B)^3-2\tr\B^2\tr\B \right){\bf I}=0,
\end{equation*}
that is,
\begin{equation*}
  \B^3 -I_2\B^2+\frac{I_2^2-I_4}{2}\B-\frac{I_2^3-2I_4I_2}{9}{\bf I}=0.
\end{equation*}
When the representing matrix of $\B$ is diagonal, its diagonal elements must satisfy a cubic equation
\begin{equation}\label{z2}
  x^3 -I_2x^2+\frac{I_2^2-I_4}{2}x-\frac{I_2^3-2I_4I_2}{9}=0.
\end{equation}
Clearly, this is the characteristic polynomial of $\B$. Moreover, there is no factor decomposition in the cubic function in \eqref{z2}.

We know that three diagonal elements $B_{11},B_{22},B_{33}$ of the diagonal matrix $B_{ij}$ satisfy equations \eqref{z1} and \eqref{z2} simultaneously, i.e.,
\begin{equation*}
\left\{\begin{aligned}
  & 2x^2-I_2x-\frac{2I_4-I_2^2}{3} = 0, \\
  & x^3 -I_2x^2+\frac{I_2^2-I_4}{2}x-\frac{I_2^3-2I_4I_2}{9}=0.
\end{aligned}\right.
\end{equation*}
That is to say, the above system has common roots. By the resultant theory in Algebra, this system of two polynomials has a common root if and only if its resultant vanishes:
\begin{equation*}
  \det\left(
        \begin{array}{ccccc}
          2 & -I_2 & -\frac{2I_4-I_2^2}{3} & 0 & 0 \\
          0 & 2 & -I_2 & -\frac{2I_4-I_2^2}{3} & 0 \\
          0 & 0 & 2 & -I_2 & -\frac{2I_4-I_2^2}{3} \\
          1 & -I_2 & \frac{I_2^2-I_4}{2} & -\frac{I_2^3-2I_4I_2}{9} & 0 \\
          0 & 1 & -I_2 & \frac{I_2^2-I_4}{2} & -\frac{I_2^3-2I_4I_2}{9} \\
        \end{array}
      \right) = 0.
\end{equation*}
By direct calculations, the resultant is indeed
\begin{equation*}
  \frac{1}{162}(I_2^2-3I_4)(I_2^2-2I_4)^2 = 0.
\end{equation*}
Hence, we only need to consider the following two cases.

\begin{itemize}
  \item When $I_2^2-3I_4=0$, the characteristic polynomial of $\B$ could be rewritten as
    \begin{equation*}
      x^3-I_2x^2+\frac{I_2^2}{3}x-\frac{1}{27}I_2^3 = \left(x-\frac{I_2}{3}\right)^3 = 0.
    \end{equation*}
    Hence $B_{11}=B_{22}=B_{33}=\frac{1}{3}I_2$, which is exactly Case II.2.2.
  \item When $I_2^2-2I_4=0$, the characteristic polynomial of $\B$ reduces to
    \begin{equation*}
      x^3-I_2x^2+\frac{I_2^2}{4}x = x\left(x-\frac{I_2}{2}\right)^2 = 0.
    \end{equation*}
    Hence $B_{11}=0$, $B_{22}=B_{33}=\frac{1}{2}I_2$. This is exactly Case II.2.1.
\end{itemize}

Hence, discussion of Case II is complete. Next, we study the last case.

\medskip

\textbf{Case III: $\D\ne0$.}
%
Since $\D$ is a second order symmetric and traceless tensor, we choose a proper positive oriented orthonormal basis $\{\e_1,\e_2,\e_3\}$ such that
\begin{equation*}
  \D = D_{11}\e_1\otimes\e_1+D_{22}\e_2\otimes\e_2-(D_{11}+D_{22})\e_3\otimes\e_3.
\end{equation*}
Clearly, if $D_{11}= D_{22}= -D_{11}-D_{22}$, we have $D_{11}=D_{22}=0$ which contradicts the assumption $\D\ne0$. Whereafter, we only need to consider the following two subcases:
\begin{itemize}
  \item[--] (III.A) $D_{11}\neq D_{22}$, $D_{11}\neq -D_{11}-D_{22}$, and $D_{22}\neq -D_{11}-D_{22}$,
  \item[--] (III.B) $D_{11}\ne D_{22}= -D_{11}-D_{22}$,
\end{itemize}
for recovering the third order symmetric and traceless tensor $\A$.

%


(III.A) Equations $E_{ij}=A_{ik\ell}\epsilon_{jm\ell}D_{km}$ are represented as
\begin{equation}\label{eqn-AD}
\left\{\begin{aligned}
  & (D_{11}+2D_{22})A_{123} = E_{11}, \\
  & (D_{11}+2D_{22})A_{223} = E_{21}, \\
  & -(D_{11}+2D_{22})(A_{112}+A_{222}) = E_{31}, \\
  & -(2D_{11}+D_{22})A_{113} = E_{12}, \\
  & -(2D_{11}+D_{22})A_{123} = E_{22}, \\
  & (2D_{11}+D_{22})(A_{111}+A_{122}) = E_{32}, \\
  & (D_{11}-D_{22})A_{112} = E_{13}, \\
  & (D_{11}-D_{22})A_{122} = E_{23}, \\
  & (D_{11}-D_{22})A_{123} = E_{33}.
\end{aligned}\right.
\end{equation}
By the assumption of (III.A), we have $D_{11}+2D_{22}\ne0$, $2D_{11}+D_{22}\ne0$, and $D_{11}-D_{22}\ne0$. Hence, seven independent elements $A_{111},A_{122},A_{112},A_{222},A_{113},A_{223},$ and $A_{123}$ are all solvable from \eqref{eqn-AD}.

(III.B) Suppose $D_{11}\ne D_{22}=-D_{11}-D_{22}$. Since the nonzero tensor $\D$ is traceless, we have $D_{11}=-2D_{22}$, $D_{22}\ne0$, and hence $\D=D_{22}{\bf d}_0$. Now, we solve \eqref{eqn-AD} for
\begin{equation*}
\left\{\begin{aligned}
  & A_{111} = -\frac{E_{32}-E_{23}}{3D_{22}}, \\
  & A_{122} = -\frac{E_{23}}{3D_{22}}, \\
  & A_{112} = -\frac{E_{13}}{3D_{22}}, \\
  & A_{113} = \frac{E_{12}}{3D_{22}}, \\
  & A_{123} = \frac{E_{22}}{3D_{22}}.
\end{aligned}\right.
\end{equation*}
Equations $A_{ijk}B_{jk}=c_i$ for $i\in\{2,3\}$ are written as
\begin{equation*}
\left\{\begin{aligned}
  & (B_{22}-B_{33})A_{222}+2B_{23}A_{223}=c_2+(B_{33}-B_{11})A_{112}-2B_{12}A_{122}-2B_{13}A_{123}, \\
  & -2B_{23}A_{222}+(B_{22}-B_{33})A_{223}=c_3+(B_{33}-B_{11})A_{113}-2B_{12}A_{123}+2B_{13}(A_{111}+A_{122})+2B_{23}A_{112}.
\end{aligned}\right.
\end{equation*}
Clearly, the determinant of a coefficient matrix of the above linear system in $A_{222}$ and $A_{223}$ is $(B_{22}-B_{33})^2+4B_{23}^2\ge0$.

(III.B.1) If $B_{22}-B_{33}\ne0$ or $B_{23}\ne0$, we solve $A_{222}$ and $A_{223}$ from the above linear system.

(III.B.2) Otherwise, we suppose $B_{22}=B_{33}$ and $B_{23}=0$. Equations $B_{33}-B_{22}=0$ and $B_{23}=0$ reduce to
\begin{equation}\label{AAA-3}
\left\{\begin{aligned}
  & A_{112}A_{222}+A_{113}A_{223} = -A_{111}^2-2A_{111}A_{122}-A_{113}^2, \\
  & A_{113}A_{222}-A_{112}A_{223} = 2A_{111}A_{123}-2A_{112}A_{113}.
\end{aligned}\right.
\end{equation}
The determinant of the above linear system in $A_{222}$ and $A_{223}$ is $-A_{112}^2-A_{113}^2\le0$.

(III.B.2.1) If $A_{112}\ne0$ or $A_{113}\ne0$, we solve $A_{222}$ and $A_{223}$ from the above linear system.

(III.B.2.2) Otherwise, we assume $A_{112}=A_{113}=0$. Then, the system \eqref{AAA-3} reduces to
\begin{equation*}
\left\{\begin{aligned}
  & A_{111}(A_{111}+2A_{122})=0, \\
  & 2A_{111}A_{123} =0.
\end{aligned}\right.
\end{equation*}

(III.B.2.2.1) If $A_{111}\ne0$, we have $A_{111}=-2A_{122}$ and $A_{123}=0$. From the equation $A_{2jk}A_{2jk}=B_{22}$, we have
\begin{equation}\label{AAA-10}
  A_{222}^2+A_{223}^2 = \frac{1}{2}B_{22}-A_{122}^2.
\end{equation}
Thus $\A={\bf d}_1(A_{222},A_{223},A_{122})$.

Now, we turn to $A_{ijk}u_k=F_{ij}$ which implies
\begin{equation*}
\left\{\begin{aligned}
  & u_2A_{222}+u_3A_{223} = F_{22}-u_1A_{122}, \\
  & -u_3A_{222}+u_2A_{223} = F_{23}.
\end{aligned}\right.
\end{equation*}
If $u_2^2+u_3^2\ne0$, we solve the above linear system and obtain $A_{222}$ and $A_{223}$ immediately.
Otherwise, we have $u_2=u_3=0$ and hence $\uu=u_1\e_1$.

When $\vv\ne0$ and $\cc\ne0$, we process a similar discussion using $A_{ijk}v_k=G_{ij}$ and $A_{ijk}c_k=K_{ij}$, respectively.

Whereafter, we consider the case that $\uu,\vv,\cc$ are parallel to $\e_1$. Since $\D=D_{22}{\bf d}_0$, $\A={\bf d}_1(A_{222},A_{223},A_{122})$ satisfying \eqref{AAA-10}, we may set
\begin{equation*}
  A_{222} = \sqrt{\frac{1}{2}B_{22}-A_{122}^2} \qquad\text{ and }\qquad A_{223}=0
\end{equation*}
by Proposition \ref{Prop-2}.

(III.B.2.2.2) Suppose $A_{111}=0$. Equations on $B_{12}$ and $B_{13}$ reduce to
\begin{equation*}
\left\{\begin{aligned}
  & A_{122}A_{222}+A_{123}A_{223} = \frac{1}{2}B_{12}, \\
  & -A_{123}A_{222}+A_{122}A_{223} = \frac{1}{2}B_{13}.
\end{aligned}\right.
\end{equation*}
The determinant of the above linear system in $A_{222}$ and $A_{223}$ is $A_{122}^2+A_{123}^2\ge0$.
If $A_{122}\ne0$ or $A_{123}\ne0$, we solve $A_{222}$ and $A_{223}$ from the above linear system.
Otherwise, we assume $A_{122}=A_{123}=0$. Moreover, the equation $A_{2jk}A_{2jk}=B_{22}$ means
\begin{equation}\label{AAA-11}
  A_{222}^2+A_{223}^2 = \frac{1}{2}B_{22}.
\end{equation}
Hence, $\A={\bf d}_1(A_{222},A_{223},0)$. 

Next, we consider $A_{ijk}u_k=F_{ij}$, which yields
\begin{equation*}
\left\{\begin{aligned}
  & u_2A_{222}+u_3A_{223} = F_{22}, \\
  & -u_3A_{222}+u_2A_{223} = F_{23}.
\end{aligned}\right.
\end{equation*}
If $u_2^2+u_3^2\ne0$, we solve the above linear system and obtain $A_{222}$ and $A_{223}$ straightforwardly.
Otherwise, we have $u_2=u_3=0$ and hence $\uu=u_1\e_1$.

When $\vv\ne0$ and $\cc\ne0$, we process a similar discussion using $A_{ijk}v_k=G_{ij}$ and $A_{ijk}c_k=K_{ij}$, respectively.

Finally, we consider the case that $\uu,\vv,\cc$ are all parallel to $\e_1$. Moreover, it holds that $\A={\bf d}_1(A_{222},A_{223},0)$ satisfying \eqref{AAA-11} and $\D=D_{22}{\bf d}_0$. By Proposition \ref{Prop-2}, we may set
\begin{equation*}
  A_{222} = \sqrt{\frac{1}{2}B_{22}} \qquad\text{ and }\qquad A_{223}=0.
\end{equation*}

In Sum, we establish the following theorem.

\begin{Theorem}\label{Thm:3Ddev}
  The $\SO$-orbit of the piezoelectric tensor $\Pie=(P_{ijk})$ is determined by a group of tensors
  \begin{equation}\label{FBtens}
    \D,~~ \B,~~ \F,~~ \G,~~ \H,~~ \uu,~~ \vv,~~  \ww,\text{ and } \cc.
  \end{equation}
\end{Theorem}

\section{A polynomially irreducible functional basis of piezoelectric tensors}\label{Sec:Func2}

According to Theorem \ref{Thm:3Ddev}, we consider 9 intermediate tensors: four vectors $\uu,\vv,\ww,\cc$, four second order symmetric and traceless tensors $\D,\H,\F,\G$, and a second order symmetric tensor $\B$. Using the approach of Smith \cite{Sm-71} and Zheng \cite{Zh-93b}, we directly obtain a set of 393 hemitropic invariants which constitute a functional basis of these 9 intermediate tensors. Since the $\SO$-orbit of the piezoelectric tensor $\Pie$ is determined by these 9 intermediate tensors, the set of 393 hemitropic invariants also form a functional basis of the piezoelectric tensor.
Furthermore, because elements of 9 intermediate tensors are polynomials of 18 independent elements of the piezoelectric tensor, these 393 hemitropic invariants may polynomially reducible, i.e., some hemitropic invariants may be polynomials in the others. With computations by Mathematica, we eliminate all the hemitropic invariants that can be polynomially represented by the others; See the Supporting Material for details. Finally, we get a polynomially irreducible functional basis of piezoelectric tensors, which contains 260 hemitropic invariants. We conclude this result in the following theorem.

\begin{Theorem}
  A functional basis of piezoelectric tensors has 260 hemitropic invariants presented in Table \ref{HemiInv}. In addition, these 260 hemitropic invariants are polynomially irreducible.
\end{Theorem}

\begin{table}
\caption{A polynomially irreducible functional basis of hemitropic invariants of piezoelectric tensors.}\label{HemiInv}
\centering
\begin{tabular}{|c|l|c|}
  \hline
  Degree & Invariants & Number \\
  \hline
 2 & $I_2:=A_{ijk}A_{ijk}$,\quad $\uu\cdot\uu$,\quad $\vv\cdot\vv$,\quad $\uu\cdot\vv$,\quad $\tr\Dd^2,$&5 \\
  \hline
   3 & $\uu\cdot\ww$,\quad $\vv\cdot\ww$,\quad $\tr \Dd^3$,\quad $\tr\D\B$,\quad $\uu\cdot\Dd\uu$,\quad $\vv\cdot\Dd\vv$,\quad $\uu\cdot\Dd\vv,$& 7 \\
  \hline
  4 & $I_4:=B_{ij}B_{ij}$,\quad $\ww\cdot\ww$,\quad $\uu\cdot\cc$,\quad $\vv\cdot\cc$\quad $[\uu,\vv,\ww]$,\quad $\tr \H^2$,\quad$\tr \F^2$,\quad $\tr \G^2$,  & 27 \\
  & $\tr \H\F$,\quad $\tr \H\G$, \quad$\tr \F\G$, \quad$\tr \Dd^2\H$,\quad$\tr \Dd^2\F$,\quad$\tr \Dd^2\G$,\quad$\uu\cdot\H\uu$,&\\
  &$\vv\cdot\H\vv$,\quad$\uu\cdot\F\uu$,\quad$\vv\cdot\F\vv$,\quad$\uu\cdot\G\uu$,\quad$\vv\cdot\G\vv$,\quad$\uu\cdot\D^2\uu,$\quad $\vv\cdot\D^2\vv,$& \\
 &$\uu\cdot\Epsilon[\D\H]$,\quad$\uu\cdot\Epsilon[\D\G]$,\quad$\vv\cdot\Epsilon[\D\H]$,\quad$[\uu,\vv,\D\uu],$\quad$[\uu,\vv,\D\vv],$&\\
  \hline
 5&$\ww\cdot\cc$,\quad$[\uu,\vv,\cc]$,\quad$\tr\D\H^2$,\quad$\tr\D\F^2$,\quad$\tr\D\G^2$,\quad$\tr\D\H\F$,\quad$\tr\D\H\G$,&35\\
  &$\tr\D\H\B$,\quad$\tr\D\F\G$,\quad$\tr\D\F\B$,\quad$\tr\D\G\B$,\quad$\ww\cdot\D\ww$,\quad$\uu\cdot\Epsilon[\B\H]$,& \\
  &$\uu\cdot\Epsilon[\H\G]$,\quad$\uu\cdot\Epsilon[\F\G]$,\quad$\vv\cdot\Epsilon[\B\H]$,\quad$\vv\cdot\Epsilon[\F\G]$,\quad$\ww\cdot\Epsilon[\D\F]$,&\\
  &$\ww\cdot\Epsilon[\D\G]$,\quad$\uu\cdot\Epsilon[\D^2\H]$,\quad$\uu\cdot\Epsilon[\D^2\F]$,\quad$\uu\cdot\Epsilon[\D^2\G]$,\quad$\vv\cdot\Epsilon[\D^2\H],$&\\
  &$\vv\cdot\Epsilon[\D^2\F],$\quad$\vv\cdot\Epsilon[\D^2\G],$\quad$\uu\cdot\F\ww$,\quad$\uu\cdot\G\ww$,\quad$\vv\cdot\G\ww$,\quad$[\uu,\vv,\H\uu]$,&\\
  &$[\uu,\vv,\F\uu]$,\quad$[\uu,\vv,\G\uu]$,\quad$[\uu,\ww,\D\uu]$,\quad$[\vv,\ww,\D\vv]$,\quad$[\uu,\vv,\H\vv]$,&\\
  &$[\uu,\vv,\G\vv]$,&\\
  \hline
  6&$I_6:=c_ic_i$,\quad$[\uu,\ww,\cc],\quad[\vv,\ww,\cc],\quad\tr\H^3,\quad\tr\F^3,\quad\tr\G^3,\quad\tr\H^2\F,$&65\\
  &$\tr\H^2\G,\quad\tr\H^2\B,\quad\tr\F^2\G,\quad\tr\H\F^2,\quad\tr\H\G^2,\quad\tr\H\B^2,\quad\tr\F\G^2,$&\\
  &$\tr\F\B^2,\quad\tr\G\B^2,\quad\tr\D^2\H^2,\quad\tr\D^2\F^2,\quad\tr\D^2\G^2,\quad\tr\H\F\G,$&\\
  &$\ww\cdot\B\ww,\quad\ww\cdot\H\ww,\quad\ww\cdot\F\ww,\quad\ww\cdot\G\ww,\quad\uu\cdot\H^2\uu,\quad\vv\cdot\H^2\vv,$&\\
  &$\uu\cdot\F^2\uu,\quad\vv\cdot\F^2\vv,\quad\vv\cdot\G^2\vv,\quad\uu\cdot\B^2\uu,\quad\vv\cdot\B^2\vv,\quad\ww\cdot\D^2\ww,$&\\
  &$[\uu,\D\uu,\D^2\uu],\quad[\vv,\D\vv,\D^2\vv],\quad\ww\cdot\Epsilon[\H\F],\quad\ww\cdot\Epsilon[\H\G],\quad\ww\cdot\Epsilon[\F\G],$&\\
  &$\cc\cdot\Epsilon[\D\F],\quad\cc\cdot\Epsilon[\D\G],\quad\ww\cdot\Epsilon[\D^2\B],\quad\ww\cdot\Epsilon[\D^2\F],\quad\ww\cdot\Epsilon[\D^2\G],$&\\
  &$\uu\cdot\Epsilon[\D\H^2],\quad\uu\cdot\Epsilon[\D\F^2],\quad\uu\cdot\Epsilon[\D\G^2],\quad\vv\cdot\Epsilon[\D\H^2],\quad\vv\cdot\Epsilon[\D\F^2],$&\\
  &$\vv\cdot\Epsilon[\D\G^2],\quad[\uu,\D\uu,\B\uu],\quad[\uu,\D\uu,\H\uu],\quad[\uu,\D\uu,\F\uu],\quad[\uu,\D\uu,\G\uu],$&\\
  &$[\vv,\D\vv,\F\vv],\quad[\vv,\D\vv,\G\vv],\quad\vv\cdot\F\cc,\quad[\uu,\ww,\B\uu],\quad[\uu,\ww,\G\uu],$&\\
  &$[\vv,\D\vv,\B\vv],\quad[\vv,\D\vv,\H\vv],\quad[\vv,\ww,\B\vv],\quad[\vv,\ww,\F\vv],\quad[\uu,\cc,\D\uu],$&\\
  &$[\vv,\cc,\D\vv],\quad[\uu,\ww,\D\ww],\quad[\vv,\ww,\D\ww],$&\\
  \hline
  7&$\cc\cdot\D\cc,\quad\cc\cdot\Epsilon[\F\G],\quad\uu\cdot\Epsilon[\B^2\H],\quad\uu\cdot\Epsilon[\B^2\F],\quad\uu\cdot\Epsilon[\B^2\G],$&54\\
  &$\uu\cdot\Epsilon[\H^2\F],\quad\uu\cdot\Epsilon[\H^2\G],\quad\uu\cdot\Epsilon[\F^2\G],\quad\vv\cdot\Epsilon[\B^2\H],\quad\vv\cdot\Epsilon[\B^2\G],$&\\
  &$\vv\cdot\Epsilon[\H^2\F],\quad\vv\cdot\Epsilon[\H^2\G],\quad\vv\cdot\Epsilon[\F^2\G],\quad\cc\cdot\Epsilon[\D^2\B],\quad\cc\cdot\Epsilon[\D^2\H],$&\\
  &$\cc\cdot\Epsilon[\D^2\F],\quad\cc\cdot\Epsilon[\D^2\G],\quad\uu\cdot\Epsilon[\B\H^2],\quad\uu\cdot\Epsilon[\B\F^2],\quad\uu\cdot\Epsilon[\B\G^2],$&\\
  &$\uu\cdot\Epsilon[\H\G^2],\quad\vv\cdot\Epsilon[\B\H^2],\quad\vv\cdot\Epsilon[\B\F^2],\quad\vv\cdot\Epsilon[\B\G^2],\quad\vv\cdot\Epsilon[\H\F^2],$&\\
  &$\vv\cdot\Epsilon[\F\G^2],\quad\ww\cdot\Epsilon[\D\B^2],\quad\ww\cdot\Epsilon[\D\H^2],\quad\ww\cdot\Epsilon[\D\F^2],\quad\ww\cdot\Epsilon[\D\G^2],$&\\
  &$[\uu,\B\uu,\H\uu],\quad[\uu,\B\uu,\F\uu],\quad[\uu,\B\uu,\G\uu],\quad[\uu,\H\uu,\F\uu],$&\\
  &$[\uu,\H\uu,\G\uu],\quad[\uu,\F\uu,\G\uu],\quad[\vv,\B\vv,\H\vv],\quad[\vv,\B\vv,\F\vv],$&\\
  &$[\vv,\B\vv,\G\vv],\quad[\vv,\H\vv,\F\vv],\quad[\vv,\H\vv,\G\vv],\quad[\vv,\F\vv,\G\vv],\quad\ww\cdot\F\cc,$&\\
  &$\ww\cdot\G\cc,\quad[\uu,\cc,\H\uu],\quad[\vv,\cc,\H\vv],\quad[\uu,\ww,\B\ww],\quad[\uu,\ww,\H\ww],$&\\
  &$[\uu,\ww,\F\ww],\quad[\uu,\ww,\G\ww],\quad[\vv,\ww,\B\ww],\quad[\vv,\ww,\H\ww],\quad[\vv,\ww,\F\ww],$&\\
  &$[\vv,\ww,\G\ww],$&\\
\end{tabular}
\end{table}

\begin{table}[t]
\centering
\begin{tabular}{|c|l|c|}
  \hline
  Degree & Invariants & Number \\
  \hline
  8&$\tr\H^2\F^2,\quad\tr\H^2\G^2,\quad\tr\H^2\B^2,\quad\cc\cdot\H\cc,\quad\cc\cdot\F\cc,\quad\cc\cdot\G\cc,$&23\\
  &$\cc\cdot\D^2\cc,\quad\ww\cdot\H^2\ww,\quad\ww\cdot\Epsilon[\B^2\F],\quad\ww\cdot\Epsilon[\B^2\G],\quad\ww\cdot\Epsilon[\H^2\F],$&\\
  &$\ww\cdot\Epsilon[\H^2\G],\quad\ww\cdot\Epsilon[\F^2\G],\quad\ww\cdot\Epsilon[\B\H^2],\quad\ww\cdot\Epsilon[\B\F^2],\quad\ww\cdot\Epsilon[\B\G^2],$&\\
  &$\ww\cdot\Epsilon[\F\G^2],\quad\cc\cdot\Epsilon[\D\H^2],\quad\cc\cdot\Epsilon[\D\F^2],\quad\cc\cdot\Epsilon[\D\G^2],\quad[\ww,\cc,\D\ww]$&\\
  &$[\uu,\cc,\D\cc],\quad[\vv,\cc,\D\cc],$&\\
  \hline
  9&$[\uu,\B\uu,\B^2\uu],\quad[\uu,\F\uu,\F^2\uu],\quad[\uu,\G\uu,\G^2\uu],\quad[\vv,\B\vv,\B^2\vv],\quad$&23\\
  &$[\vv,\G\vv,\G^2\vv],\quad[\ww,\D\ww,\D^2\ww],\quad\cc\cdot\Epsilon[\B^2\F],\quad\cc\cdot\Epsilon[\B^2\G],$&\\
  &$\cc\cdot\Epsilon[\H^2\F],\quad\cc\cdot\Epsilon[\H^2\G],\quad\cc\cdot\Epsilon[\B\H^2],\quad\cc\cdot\Epsilon[\B\F^2],\quad\cc\cdot\Epsilon[\B\G^2],$&\\
  &$[\ww,\D\ww,\B\ww],\quad[\ww,\D\ww,\H\ww],\quad[\ww,\D\ww,\F\ww],\quad[\ww,\D\ww,\G\ww],$&\\
  &$[\ww,\cc,\B\ww],\quad[\ww,\cc,\H\ww],\quad[\ww,\cc,\F\ww],\quad[\ww,\cc,\G\ww],$&\\
  &$[\uu,\cc,\G\cc],\quad[\vv,\cc,\F\cc],$&\\
  \hline
 10&$I_{10}:=A_{ijk}c_ic_jc_k,\quad[\ww,\B\ww,\H\ww],\quad[\ww,\B\ww,\F\ww],\quad[\ww,\B\ww,\G\ww],$&10\\
 &$[\ww,\H\ww,\F\ww],\quad[\ww,\H\ww,\G\ww],\quad[\ww,\F\ww,\G\ww],\quad[\ww,\cc,\B\cc],$&\\
 &$[\ww,\cc,\F\cc],\quad[\ww,\cc,\G\cc],\quad$&\\
 \hline
 12&$[\ww,\B\ww,\B^2\ww],\quad[\cc,\D\ww,\B\cc],\quad[\cc,\D\ww,\H\cc],\quad[\cc,\D\cc,\F\cc],$&5\\
 &$[\cc,\D\cc,\G\cc],$&\\
 \hline
 13&[\cc,\B\cc,\H\cc],\quad[\cc,\B\cc,\F\cc],\quad[\cc,\B\cc,\G\cc],\quad[\cc,\H\cc,\F\cc],&5\\
 &$[\cc,\H\cc,\G\cc],$&\\
 \hline
 15&$[\cc,\B\cc,\B^2\cc]$.&1\\
 \hline
 Total&&260\\
 \hline
\end{tabular}
\end{table}

In the remainder of this section, we compare our result with some existing  works.
First, we consider a special case that the piezoelectric tensor is a third order symmetric and traceless tensor, i.e., $\D=0$, $\uu=\vv=0$ and hence $\Pie=\A\in\SH{3}$. According to Olive and Auffray \cite{OA-14}, a set of five hemitropic invariants with degrees two, four, six, ten, and fifteen forms a minimal integrity basis of third order symmetric and traceless tensors.
For our result, there are only five nonzero hemitropic invariants from Table \ref{HemiInv} in this case:
\begin{equation}\label{dev3-inb}
  I_2,\quad I_4,\quad I_6,\quad I_{10}, \quad\text{ and }\quad   [\cc,\B\cc,\B^2\cc].
\end{equation}
Clearly, degrees of these five polynomially irreducible hemitropic invariants are two, four, six, ten, and fifteen, respectively.
Moreover, the set of five hemitropic invariants in \eqref{dev3-inb} and the one of Olive and Auffray \cite{OA-14} are equivalent.

Second, a third order symmetric tensor is also a special piezoelectric tensor. Also in \cite{OA-14}, there are 27 hemitropic invariants consisting of a minimal integrity basis of third order symmetric tensors. For convenience, we use the same symbols as \cite{OA-14} and list all these 27 hemitropic invariants $\{i_{2},\ j_{2},\dots,i_{15}\}$ of different degrees in Table \ref{HemiInv2}.

\begin{table}[t]
  \centering
  \caption{A minimal integrity basis of third order symmetric tensors has 27 hemitropic invariants.}\label{HemiInv2}
  \begin{tabular}{|c|l|c|l|}
    \hline
   Degree &Hemitropic Invariants  & Degree&  Hemitropic Invariants \\
\hline
2&$i_2,~~j_2,$&9&$i_9,~~j_9,~~k_9,~~l_9,~~m_9,~~n_9,~~o_9,$\\
\hline
4&$i_4,~~j_4,~~k_4,~~l_4,$&10&$i_{10},$\\
\hline
6&$i_6,~~j_6,~~k_6,~~l_6,~~m_6,$&11&$i_{11},~~j_{11},$\\
\hline
7&$i_7,~~j_7,~~k_7,$&13&$i_{13},$\\
\hline
8&$i_8,$&15&$i_{15}.$\\
    \hline
  \end{tabular}
\end{table}

In this case, we have $\D={\bf 0}$, $\vv=0$, and
$$P_{ijk}=A_{ijk}+\frac{1}{5}(u_i\delta_{jk}+u_j\delta_{ik}+u_k\delta_{ij})\in \mathbb{T}_{(ijk)}.$$
Thus, $\Pie$ only refers to $\B,\cc,\uu,$ and $\F$. From Table \ref{HemiInv}, there are exactly 20 nonzero hemitropic invariants in Table \ref{HemiInv3}, which form a functional basis of third order symmetric tensors. So we have 7 less hemitropic invariants than Olive and Auffray's basis.

\begin{table}
  \centering
  \caption{A functional basis of third order symmetric tensors contains 20 hemitropic invariants.}\label{HemiInv3}
  \begin{tabular}{|c|l|c|}
    \hline
   Degree & Hemitropic Invariants &Number  \\
\hline
2&$I_2,~~\uu\cdot\uu,$&2\\
\hline
4&$I_4,~~\uu\cdot\vt{c},~~\mathrm{tr}\F^2,~~\uu\cdot\F\uu,$&4\\
\hline
6&$I_6,\ \mathrm{tr}\F^3,\ \mathrm{tr}\F\B^2,\ \uu\cdot\F^2\uu,\ \uu\cdot\B^2\uu,$ &5\\
\hline
7&$\uu\cdot\Epsilon[\B\F^2],~~[\uu,\B\uu,\F\uu],$&2\\
\hline
8&$\vt{c}\cdot\F\vt{c},$&1\\
\hline
9&$[\uu,\B\uu,\B^2\uu],~~[\uu,\F\uu,\F^2\uu],~~\cc\cdot\Epsilon[\B^2\F],~~\cc\cdot\Epsilon[\B\F^2],$&4\\
\hline
10&$I_{10},$&1\\
\hline
15&$[\cc,\B\cc,\B^2\cc].$&1\\
    \hline
  \end{tabular}
\end{table}

On one hand, since $\{i_{2},\ j_{2},\dots,i_{15}\}$ forms a minimal integrity basis, the above 20 hemitropic invariants $\{I_2, \dots,[\cc,\B\cc,\B^2\cc]\}$ can absolutely be polynomial represented by $\{i_{2},\ j_{2},\dots,i_{15}\}$. On the other hand, because we have seven less hemitropic invariants, there exist some invariants in $\{i_{2},\ j_{2},\dots,i_{15}\}$ that are not polynomials in $I_2, \dots,[\cc,\B\cc,\B^2\cc]$. It means that our functional basis is not a subset of Olive's minimal integrity basis. With further calculations, we find all the possible polynomial representations of $\{i_{2},\ j_{2},\dots,i_{15}\}$ as Table \ref{PolyR}. Here, ``$a\sim b\ \oplus\ c$" means $a$ can be linear represented by $b$ and $c$. We notice that $i_7,\ k_9,\ l_9,\ m_9,\ n_9,\ i_{11},\ j_{11},\ i_{13}$ can not be polynomial represented by $\{I_2, \dots,[\cc,\B\cc,\B^2\cc]\}$, which agrees with our inference.

\begin{table}
\centering
\caption{Polynomial relations.}\label{PolyR}
\begin{tabular}{|c|l|}
\hline
Degree&Polynomial Relations\\
\hline
2&$i_{2}\sim(\mathrm{tr}\B),\ j_2\sim(\uu\cdot\uu),$\\
\hline
4&$i_4\sim(\mathrm{tr}\B)^2\oplus(\mathrm{tr}\B^2),\quad j_4\sim(\uu\cdot\vt{c}),\quad k_4\sim(\mathrm{tr}\F^2)\oplus((\mathrm{tr}\B)(\uu\cdot\vt{c})),$\\
&$l_4\sim(\uu\cdot\F\uu),\quad i_6\sim(\vt{c}\cdot\vt{c}),\quad j_6\sim(\mathrm{tr}\F\B^2)\oplus((\mathrm{tr}\B)(\uu\cdot\vt{c}))$,\\
\hline
6&$k_6\sim((\mathrm{tr}\B)(\mathrm{tr}\F^2))\oplus((\mathrm{tr}\B)^2(\uu\cdot\vt{c}))\oplus((\mathrm{tr}\B)^2(\uu\cdot\vt{c}))\oplus(\uu\cdot\B^2\uu),$\\
&$l_6\sim((\mathrm{tr}\B)(\uu\cdot\F\uu))\oplus((\uu\cdot\uu)(\uu\cdot\vt{c}))
\oplus\mathrm{tr}\F^3,$\\
&$m_6\sim((\mathrm{tr}\B)(\uu\cdot\uu))\oplus((\uu\cdot\uu)(\mathrm{tr}\F^2))
\oplus(\uu\cdot\F^2\uu)$,\\
\hline
7
&$j_7\sim(\uu\cdot\Epsilon[\B\F^2]),\quad k_7\sim([\uu,\B\uu,\F\uu]),$\\
\hline
8&$i_8\sim(\vt{c}\cdot\F\vt{c}),$\\
\hline
9&
$i_9\sim(\cc\cdot\Epsilon[\B^2\F]),\quad j_9\sim(\cc\cdot\Epsilon[\B\F^2]),\quad o_9\sim[(\uu\cdot\uu)([\uu,\B\uu,\F\uu])]\oplus([\uu,\F\uu,\F^2\uu]),$\\
\hline
10&$i_{10}\sim(A_{ijk}c_ic_jc_k),$\\
\hline
15&$i_{15}\sim([\cc,\B\cc,\B^2\cc]).$\\
\hline
\end{tabular}
\end{table}

\begin{table}[t]
\caption{Numbers of hemitropic invariants in different degrees.}\label{Num}
\centering
\begin{tabular}{|c|c|c||c|c|c|}
\hline
Degree & FB & MIB  &Degree & FB & MIB  \\
\hline
2 & 5 & 5&9& 23 & 55 \\\hline
3 & 7& 7&10& 10 & 14\\\hline
4 & 27& 28&11& 0 & 10\\\hline
5 &35 & 45& 12&5 & 6\\\hline
6 & 65& 105&13& 5 & 1\\\hline
7 &54 & 126& 14&0 & 1\\\hline
8 &23 & 91& 15&1 & 1\\
\hline
Total & 260 & 495 & & & \\\hline
\end{tabular}
\end{table}

Finally, we list in Table \ref{Num} numbers of hemitropic invariants in different degrees of our polynomially irreducible functional basis (FB) and Olive's minimal integrity basis (MIB). Compared with Olive's result, the number of hemitropic invariants in the new functional basis is nearly a half of Olive's one.

\section{Final remarks}\label{Sec:FinRemk}

A polynomially irreducible functional basis of 260 hemitropic invariants of piezoelectric tensors has been constructed in this paper. There are 125 odd order hemitropic invariants and 135 even order hemitropic invariants in the new functional basis. We note that this polynomially irreducible functional basis of piezoelectric tensors are not necessary a minimal functional basis.

At last, we claim that functional bases of piezoelectric tensors are more complex than that of elasticity tensors. On one hand, hemitropic invariants of a piezoelectric tensor are not necessary isotropic, since the piezoelectric tensor is of odd order. Nevertheless, hemitropic invariants and isotropic invariants of even order elasticity tensors are equivalent. So we focus on hemitropic invariants that forms a functional basis of piezoelectric tensors in this paper.

On the other hand, by the orthogonal irreducible decomposition of tensors, the elasticity tensor is factorized into a fourth order symmetric and traceless tensor, two second order symmetric and traceless tensors, and two scalars. Beside two scalars which are invariants naturally, we only need to consider invariants and joint invariants of three symmetric and traceless tensors for functional bases of elasticity tensors. However, according to \eqref{decom-piez-space}, we must study invariants and joint invariants of four symmetric and traceless tensors for functional bases of piezoelectric tensors.
In this sense, the polynomially irreducible functional basis of 260 hemitropic invariants of piezoelectric tensors is also significant for the theory of representations for tensor functions.

\section*{Acknowledgments}

The authors are grateful to Professor Quanshui Zheng from Tsinghua University for his valuable comments.


\begin{thebibliography}{10}

\bibitem{Boe-77} J.P. Boehler, ``On irreducible representations for isotropic scalar functions'',
  {\sl ZAMM \bf 57} (1977) 323--327.

\bibitem{Boe-87} J.P. Boehler, {\sl Application of Tensor Functions in Solid Mechanics,} CISM Courses and Lectures, edited by J.-P. Boehler, Springer-Verlag, Wien, 1987.

\bibitem{BKO-94} J.P. Boehler, A.A. Kirillov, and E.T. Onat, ``On the polynomial invariants of the elasticity tensor'', {\sl J. Elast. \bf 34}(2) (1994) 97--110.

\bibitem{CJQ-17} Y. Chen, A. J\'akli, and L Qi, ``Spectral Analysis of Piezoelectric Tensors'', arXiv:1703.07937, 2017.

\bibitem{CHQZ-18} Y. Chen, S. Hu, L. Qi and W. Zou, ``Irreducible function bases of isotropic invariants of a third order three-dimensional symmetric and traceless tensor'', arXiv:1712.02087v7, 2018.

\bibitem{CLQZZ-18}  Z. Chen, J. Liu, L. Qi, Q.S. Zheng and W.N. Zou, ``An irreducible function basis of isotropic invariants of a third order three-dimensional symmetric tensor'', {\sl J. Math. Phys. \bf 59} (2018) 081703.

\bibitem{CC-80} J. Curie and P. Curie, ``D\'{e}veloppement, par pression, de l'\'{e}lectricit\'{e}
  polaire dans les cristaux h\'{e}mi\`{e}dres \`{a} faces inclin\'{e}es'',
  {\sl Comptes rendus (in French) \bf 91} (1880) 294--295.

\bibitem{CC-81} J. Curie and P. Curie, ``Contractions et dilatations produites par des tensions
  \'{e}lectriques dans les cristaux h\'{e}mi\`{e}dres \`{a} faces inclin\'{e}es'',
  {\sl Comptes rendus (in French) \bf 93} (1881) 1137--1140.


\bibitem{Ha-07} S. Hauss\"uhl, {\sl Physical Properties of Crystals: An Introduction}, Wiley-VCH Verlag, Weinheim, 2007.

\bibitem{Hi-93} D. Hilbert, Theory of Algebraic Invariants, Cambridge university press, Cambridge, 1993.

\bibitem{Kh-08} A.L. Kholkin, N.A. Pertsev and A.V. Goltsev, ``Piezolelectricity and crystal symmetry'', {\sl Piezoelectric and Acoustic Materials}, (2008) 17-38.

\bibitem{LDQZ-18} J. Liu, W. Ding, L. Qi, W. Zou, ``Isotropic polynomial invariants of the Hall tensor'', {\sl Appl. Math. Mech. \bf 39} (2018) 1845--1856.


\bibitem{Lo-89} D.R. Lovett, {\sl Tensor Properties of Crystals,} Second Edition, Institute of Physics Publishing,
Bristol, 1989.

\bibitem{Ny-85} J.F. Nye, {\sl Physical Properties of Crystals: Their Representation by Tensors and Matrices,} Second Edition, Clarendon Press, Oxford, 1985.

\bibitem{Ol-14} M. Olive, ``G\'{e}om\'{e}trie des espaces de tenseurs Une approche effective appliqu\'{e}e $\grave{\rm a}$ la m\'{e}canique des milieux continus." {\sl PhD thesis (in French)} (2014).


\bibitem{Ol-17} M. Olive, ``About Gordan's algorithm for binary forms'', {\sl Found. Comput. Math., \bf 17}(6) (2017) 1407--1466.

\bibitem{OA-14} M. Olive and N. Auffray, ``Isotropic invariants of completely symmetric third-order tensor'', {\sl J. Math. Phys. \bf 55} (2014) 092901.

\bibitem{OKA-17} M. Olive, B. Kolev and N. Auffray, ``A minimal integrity basis for the elasticity tensor'',
  {\sl Arch. Rational Mech. Anal. \bf 226} (2017) 1--31.

\bibitem{PT-87} S. Pennisi and M. Trovato, ``On the irreducibility of Professor G.F. Smith's representations for isotropic functions'', {\sl Int. J. Eng. Sci. \bf 25} (1987) 1059--1065.

\bibitem{Sm-71} G.F. Smith, ``On isotropic functions of symmetric tensors, skew-symmetric tensors and vectors'', {\sl Int. J. Eng. Sci., \bf 9}(10) (1971) 899--916.

\bibitem{SB-97} G.F. Smith and G. Bao, ``Isotropic invariants of traceless symmetric tensors of orders three and four'', {\sl Int. J. Eng. Sci., \bf 35} (1997) 1457--1462.

\bibitem{Spe-70} A.J.M. Spencer, ``A note on the decomposition of tensors into traceless symmetric tensors'', {\sl Int. J. Eng. Sci., \bf 8}(6) (1970) 475--481.

\bibitem{Ste-95} S. Sternberg, {\sl Group Theory and Physics}, Cambridge University Press, 1995.

\bibitem{TM-06} A. Thionnet and C. Martin, ``A new constructive method using the theory of invariants to obtain behavior laws'', {\sl Int. J. Solids. Struct, \bf 43}(2) (2006) 325-345.


\bibitem{W-701} C.C. Wang, ``A new representation theorem for isotropic functions: An answer to Professor G. F. Smith's criticism of my papers on representations for isotropic functions, Part 1. Scalar-valued isotropic functions,'' {\sl Arch. Ratl. Mech. Anal, \bf 36}(3) (1970) 166-197.

\bibitem{W-702} C.C. Wang, ``A new representation theorem for isotropic functions: An answer to Professor G. F. Smith's criticism of my papers on representations for isotropic functions, Part 2. Vector-valued isotropic functions, symmetric tensor-valued isotropic functions, and skew-symmetric tensor-valued functions,'' {\sl Arch. Ratl. Mech. Anal, \bf 36}(3) (1970) 198-223.

\bibitem{W-7071} C.C. Wang, ``Corrigendum to my recent papers on ``Representations for isotropic functions'' Vol.36,pp.166-197,198-223(1970),'' {\sl Arch. Ratl. Mech. Anal, \bf 43}(5) (1971) 392-395.

\bibitem{Zh-93} Q.S. Zheng, ``On the representations for isotropic vector-valued, symmetric tensor-valued and skew-symmetric tensor-valued functions'', {\sl Int. J. Eng. sci., \bf 31}(7), (1993) 1013--1024.

\bibitem{Zh-93b} Q.S. Zheng, ``On transversely isotropic, orthotropic and relative isotropic functions of symmetric tensors, skew-symmetric tensors and vectors. Part I: two dimensional orthotropic and relative isotropic functions and three dimensional relative isotropic functions'', {\sl Int. J. Eng. Sci., \bf 31}(10), (1993) 1399--1409.

\bibitem{Zh-94} Q.S. Zheng, ``Theory of representations for tensor functions --- a unified invariant approach to constitutive equations'', {\sl Appl. Mech. Rev., \bf 47}(11), (1994) 545--587.

\bibitem{ZTP-13} W.N. Zou, C.X. Tang, and E. Pan, ``Symmetric types of the piezotensor and their identification'', {\sl Proc. R. Soc. A \bf 469} (2013) 20120755.

\bibitem{ZZ-01} W.N. Zou, Q.S. Zheng, D.X. Du, and J. Rychlewski, ``Orthogonal irreducible decompositions of tensors of high orders'', {\sl Math. Mech. Solids \bf 6}(3) (2001) 249--267.

\end{thebibliography}
\end{document}